\documentclass[conference]{IEEEtran}

\IEEEoverridecommandlockouts

\ifCLASSINFOpdf
   \usepackage[pdftex]{graphicx}

\else

\fi
\usepackage{caption}
\usepackage{subcaption}
\usepackage{setspace}
\usepackage{graphicx}
\usepackage{epstopdf}
\usepackage{color}
\usepackage{url}
\usepackage[cmex10]{amsmath}
\usepackage{amsthm}
\usepackage{mathtools}
\usepackage{mathptmx} 
\usepackage{times} 
\usepackage{amssymb}
\usepackage{float}
\usepackage{dsfont}

\usepackage{color}

\begin{document}

\title{Communication over a Time Correlated Channel with an Energy Harvesting Transmitter}

\author{
\IEEEauthorblockN{Mehdi Salehi Heydar Abad}
\IEEEauthorblockA{\normalsize Faculty of Engineering \\
and Natural Sciences\\
Sabanci University,
Istanbul, Turkey\\ mehdis@sabanciuniv.edu}
\and
\IEEEauthorblockN{Deniz G\"und\"uz}
\IEEEauthorblockA{\normalsize Department of Electrical \\
and Electronic Engineering,\\
Imperial College London, U.K.\\ d.gunduz@imperial.ac.uk}
\and
\IEEEauthorblockN{Ozgur Ercetin}
\IEEEauthorblockA{\normalsize Faculty of Engineering \\
and Natural Sciences\\
Sabanci University,
Istanbul, Turkey\\ oercetin@sabanciuniv.edu}
\thanks{ This work was in part supported by EC H2020-MSCA-RISE-2015 pro- gramme under grant number 690893, by Tubitak under grant number 114E955 and by British Council Institutional Links Program under grant number 173605884.}
}

\maketitle
\newtheorem{theorem}{Theorem}
\newtheorem{lemma}{Lemma}
\newtheorem{corollary}{Corollary}

\begin{abstract}
In this work, communication over a time-correlated point-to-point wireless channel is studied for an energy harvesting (EH) transmitter. In this model, we take into account the time and energy cost of acquiring channel state information. At the beginning of the time slot, the EH transmitter, has to choose among three possible actions: i) deferring the transmission to save its energy for future use, ii) transmitting without sensing, and iii) sensing the channel before transmission. At each time slot, the transmitter chooses one of the three possible actions to maximize the total expected discounted number of bits transmitted over an infinite time horizon. This problem can be formulated as a partially observable Markov decision process (POMDP) which is then converted to an ordinary MDP by introducing a belief on the channel state, and the optimal policy is shown to exhibit a  threshold behavior on the belief state, with battery-dependent threshold values. Optimal threshold values  and corresponding optimal performance are characterized through numerical simulations, and it is shown that having the sensing action and intelligently using it to track the channel state improves the achievable long-term throughput significantly.
\end{abstract}

\IEEEpeerreviewmaketitle

\section{Introduction}
\vspace{-0.1cm}
Due to the tremendous increase in the number of battery-powered wireless communication devices over the past decade, replenishing the batteries of these devices by harvesting energy from natural resources has become an important research area \cite{Paradiso}. Transmitters may harvest energy via wind turbines, photovoltaic cells, thermoelectric generators, or from mechanical vibrations through piezoelectric or electromagnetic technology \cite{EHsource}. Regardless of which type of energy harvesting (EH) device and natural energy source is employed, a main concern is the stochastic nature of the EH process driving the wireless communications. The associated battery recharging process can be modeled either as a continuous, or a discrete, \cite{EHcont,EHdos} stochastic process. 

We consider a wireless point-to-point link with a transmitter equipped with a finite-capacity battery fed by an EH device.  At each time slot, a unit of energy is harvested by the transmitter according to a binary random process independent over time\footnote{Typically, the EH process is neither memoryless nor discrete, and the energy is accumulated continuously over time.  However, in order to develop the analytical model underlying this paper, we follow the common assumption in the literature \cite{EHcont}, and assume that the continuous energy arrival is accumulated in an intermediate energy storage device to form quantas.}. We assume that the transmitter can accurately observe the current energy level of the battery, and it has the knowledge of the statistics of the EH process. 
The wireless channel is time-varying and has memory across time. The channel memory is modeled with a finite state Markov chain \cite{markov}, where the next channel state depends only on the current state. A convenient and often-employed simplification of the Markov model is a two state Markov chain, known as the Gilbert-Elliot channel \cite{gilbert}. This model assumes that the channel can be either in a \emph{good} or a \emph{bad} state. We assume that in the bad state, transmitter cannot transmit any information reliably, while in the good state it may transmit $R$ bits per time slot by spending exactly one unit of energy from its battery.


 In this work, differently from most of the literature on EH systems, we take into account the energy cost of acquiring channel state information (CSI). At the beginning of each time slot, without knowing the current CSI, EH transmitter has three possible actions: i) deferring the transmission to save its energy for future use, ii) transmitting at a rate of $R$ bits per time slot, and iii) sensing the channel to reveal the channel state by consuming a portion of its energy and transmission time, followed by  transmission at a reduced rate consuming the remainder of the energy unit, if the channel is in the good state. If the channel is in a bad state, the transmitter remains silent in the rest of the time slot, saving its energy for future. If the level of the battery is less than a unit of energy at the beginning of a time slot, no transmission is possible. Our objective is to maximize the total expected discounted number of bits transmitted over an infinite time horizon.

Markov decision process (MDP) tools have been extensively utilized in the recent literature in solving communication problems involving EH devices. In \cite{MDP1} authors propose a simple single-threshold policy for a solar-powered sensor operating over a fading wireless channel. Optimality of a single-threshold policy is proven \cite{zorzi} when transmitting packets with importance values on EH transmitter. Problem of energy allocation for gathering and transmitting data in an EH communication system is studied in \cite{MDP2} and \cite{MDP3}. The scheduling of EH transmitters with time correlated energy arrivals  to optimize the long term sum throughput is investigated in \cite{deniz}. The allocation of energy over a finite horizon to optimize the throughput is considered in \cite{MDP4}, where it is assumed that either the current or the future energy and channel states are provided to the transmitter. In \cite{MDP5}, for a Markov EH process, and a static channel, a discrete power allocation problem is studied to maximize the throughput. In \cite{mehdi} throughput is optimized over a multiple access channel with collisions, considering spatially correlated energy arrivals at the transmitters. 

In a closely related work \cite{efremidus}, scheduling of an EH transmitter over a Gilbert-Elliot channel is considered. However, unlike our work, the transmitter \cite{efremidus} always has perfect CSI, obtained by sensing at every time slot, and makes a decision to defer or to transmit, based on the current CSI and battery state. Similarly, without considering the channel sensing capability, \cite{Aprem} addresses the problem of optimal power management for an EH sensor over a multi-state wireless channel with memory, using the ACK/NACK channel feedback to track the channel state. In our work, instead, we take into account the energy cost of channel sensing which can be significant for EH transmitters. Therefore, the EH transmitter does not necessarily have perfect CSI, but keeps an updated belief of the channel state according to its past observations.  Hence, the transmitter may occasionally take a third decision (in addition to defer and to transmit) of sensing the current channel state to improve its belief. Channel sensing is an essential part of opportunistic and cognitive spectrum access. In \cite{sense1}, the authors investigate the problem of optimal access to a Gilbert-Elliot channel, wherein an energy-unlimited transmitter senses the channel at every time slot. In \cite{gilbertmakale} channel sensing is done only occasionally. The transmitter can decide to transmit at a high or a low rate without sensing the channel; or can first sense the channel and transmit at a reduced rate due to the time spent for sensing. The energy cost of sensing is ignored in \cite{gilbertmakale}.

In Section \ref{sec:SystemModel} we explain the channel and EH process models under consideration, and elaborate on the transmission protocol. In Section \ref{sec:MDPFormulation}, we formulate the problem as a two state partially observable MDP (POMDP) which is then converted to a continuous-state MDP by introducing a belief state. In Section \ref{sec:StructureOfTheOptimalPolicy} we show that the optimal policy is of threshold type, for which the optimal threshold values depend on the state of the battery. In Section \ref{num} we present simulation results that numerically obtain the optimal threshold values and the optimal performance. In Section \ref{concl} we conclude the paper and present future research directions.

\begin{figure}[t]
  \centering
    \includegraphics[scale=.2]{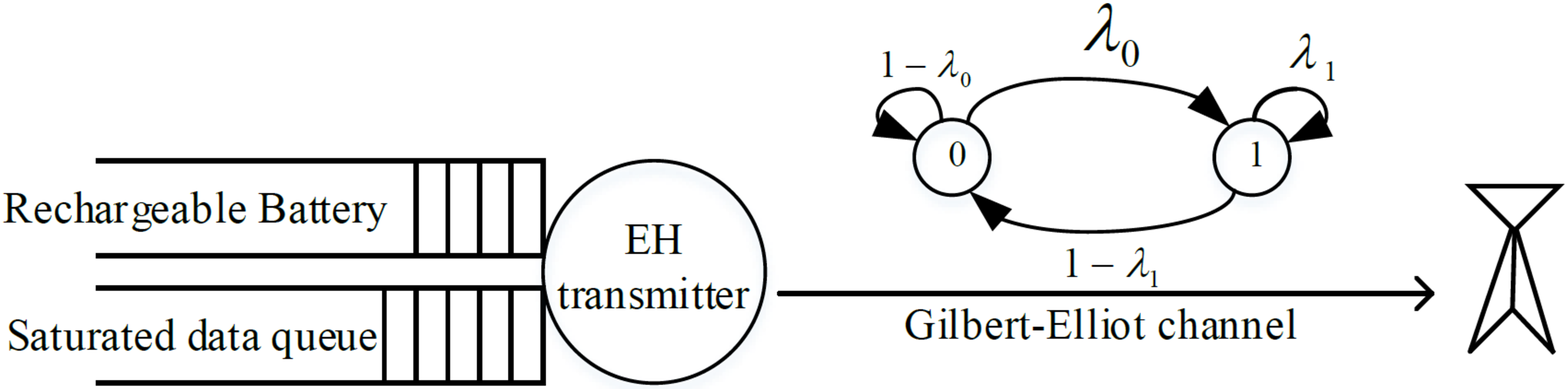}
		  \caption{System model.}
			\label{EHmodel}
\end{figure}

\section{System Model}
\label{sec:SystemModel}
\subsection{Channel and energy harvesting models}
\label{sec:ChannelModelAndEnergyHarvestingAssumptions}

Consider the communication system illustrated in Fig. \ref{EHmodel}, in which an EH transmitter communicates over a slotted Gilbert-Elliot channel. Let $G_t$ denote the state of the channel at time slot $t$ which is modeled as a one-dimensional Markov chain with two states: a good state denoted by $1$, and a bad state denoted by $0$. Channel transitions occur at the beginning of each time slot. The transition probabilities are given by $\mathds{Pr}\left[G_t=1|G_{t-1}=1\right]=\lambda_1$ and $\mathds{Pr}\left[G_t=1|G_{t-1}=0\right]=\lambda_0$. The transmitter can transmit $R$ bits per time slot if $G_t=1$, and zero bits if $G_t=0$.  

A unit of energy arrives at the end of  time slot $t$ according to an independent and identically distributed (i.i.d.) Bernoulli process, denoted by $E_t$, with probability $q$, i.e., $\mathds{P}\left[E_t=1\right]=q$ for all $t$. The transmitter stores the energy packets in a battery with a storage capacity of $B_{max}$ units of energy. We denote the state of the battery, i.e., the energy available in the battery at the beginning of time slot $t$, by $B_t$.
An energy unit is consumed at each slot if the transmitter decides to transmit in that slot. A unit of energy consumed per slot includes the energy cost of sensing (if the transmitter decides to sense the channel), transmission of the message, and the reception of ACK or NACK from the receiver. We assume that the transmitter has an infinitely backlogged data queue, and thus, it always has a packet to transmit. 

\subsection{Transmission protocol}
\label{sec:TransmissionProtocol}

At the beginning of each time slot, the transmitter may choose among three possible actions: $i)$ deferring the transmission, $ii)$ channel sensing and transmitting opportunistically, and $iii)$ transmitting without sensing.

\emph{Deferring the transmission:} 
This action (denoted by $D$) corresponds to the case in which the transmitter either believes that the channel is in a bad state, or observes that its battery has low energy. If this action is chosen, there is no message exchange between the transmitter and the receiver. Hence, the receiver does not send any feedback, and therefore the transmitter cannot obtain any knowledge about the current channel state. The scenario in which the transmitter is informed about the current channel state even when it does not transmit any data packet is equivalent to the system model investigated in \cite{efremidus}.

\emph{Channel sensing and transmitting opportunistically:} This action (denoted by $O$) corresponds to the case in which the transmitter decides to sense the channel at the beginning of the time slot. We assume that sensing consumes a fraction $0<\tau<1$ of an energy unit. Sensing is carried out by the transmitter first sending a control/probing packet, to which, the receiver responds with a packet indicating the channel state. We assume that the time it takes to sense the channel is $\tau$ seconds and the transmitter consumes on average the same power as data transmission over the sensing period. Therefore, we equivalently assume for simplicity that $\tau = 1/k$ for some $k\in \mathbb{Z}^+$. In the remaining $1-\tau$ seconds, the transmitter may choose to transmit data at the same rate it would without channel sensing, which means that by the end of the time slot it transmits $(1-\tau)R$ bits per time slot. 

If the channel is revealed to be in the bad state, transmitter defers its transmission and saves the rest of the energy unit (i.e., $1-\tau$). Note that thanks to the channel sensing capability, in the case of a bad state, the transmitter wastes only $\tau$ portion of a unit energy packet, and saves the remaining energy by deferring its transmission, which as we will show later in this paper, is an important advantage in EH networks with scarce energy sources. 

\emph{Transmitting without sensing:} This action (denoted by $T$) corresponds to the case when transmitter attempts to transmit $R$ bits in the current time slot without sensing the channel. If the channel is in a good state, the transmission is successful and the receiver sends an ACK. Otherwise, the transmission fails, and the receiver sends a NACK. Note that, at the end of the slot the transmitter has the perfect knowledge of the current channel state.


\section{Partially Observable Markov Decision Process (POMDP) formulation}
\label{sec:MDPFormulation}

At the beginning of each time slot, the transmitter chooses among the three possible actions based on the state of its battery, and its belief about the channel state to maximize a long-term discounted reward to be defined shortly. Although the transmitter is perfectly aware of its battery state, it cannot directly observe the current channel state. Hence, the problem in hand becomes a partially observable Markov decision process (POMDP).

Let the state of the system at time $t$ be denoted by $S_{t}= (B_{t},X_{t})$. We define the $\mathit{belief}$ of the transmitter at time slot $t$, denoted by $X_{t}$, as the conditional probability that the channel is in the good state at the beginning of the current slot, i.e., $X_{t}=\mathds{Pr}\left[G_t=1|\mathcal{H}_{t}\right]$, given the history $\mathcal{H}_{t}$, where $\mathcal{H}_{t}$ represents all the past actions and observations of the transmitter up to slot $t$. The transmitter's belief constitutes a sufficient statistic to characterize its optimal actions \cite{lovejoy}. Note that with this definition of the state, the POMDP problem is converted into a MDP with an uncountable state space $\left[0,\tau,2\tau,\ldots,B_{max}\right]\times\left[0,\ 1\right]$\footnote{Note that since sensing without transmission is possible, i.e., consuming only $\tau$ fraction of the energy unit, the battery can take fraction of units as states.}.

A transmission policy $\pi$ describes a set of rules that dictates which action to take depending on the history. Let $V^{\pi}(b,p)$ be the expected infinite-horizon discounted reward with initial state $S_0=(b,\ \mathds{Pr}\left[G_0=1|\mathcal{H}_{0}\right]=p)$ under policy $\pi$ with discount factor $\beta\in[0,\ 1)$. The use of the expected discounted reward allows us to obtain a tractable solution, and one can gain insights into the optimal policy for the average reward when $\beta$ is close to 1.  It is also discussed in [5] that $\beta$ can be interpreted as the probability that a particular user is allowed to use the channel, or as the probability of the transmitter to remain active at each time slot as in \cite{learningdeniz}. For an initial belief $p$ the expected discounted reward has the following expression
\begin{align}
V^{\pi}(b,\ p) = \mathds{E}\left[\sum^{\infty}_{t=0}\beta^{t}R(S_{t},A_{t})|S_{0}=(b,\ p)\right],\label{Vdef}
\end{align}
where $t$ is the time index, $A_{t}\in\left\{D,O,T\right\}$ is the action chosen at time $t$, and $R(S_{t},A_{t})$ is the expected reward acquired when action $A_t$ is taken at state $S_{t}$. The expectation in (\ref{Vdef}) is over state sequence distribution induced by the given transmission policy $\pi$. The expected reward when action $A_t$ is chosen at state $S_t$ is given as follows:
\begin{align}
R(S_{t},A_{t}) = \left\{
\begin{array}{ll}
X_{t}R & \text{if }  A_t = T\ \text{and}\ B_t\geq 1,\\
(1-\tau)X_t R & \text{if }  A_t = O\ \text{and}\ B_t\geq 1,\\
0 & \text{otherwise}.
\end{array} \right.\label{eq:R}
\end{align}
Since at least one energy unit is required for transmission, if the battery state is less than one unit, the reward becomes zero. Hence, in explaining the expected reward function in (\ref{eq:R}), we consider actions when the battery state is greater than or equal to one. If the action of transmitting without sensing is chosen, $R$ bits per time slot are transmitted successfully if the channel is in a good state, and $0$ bits if the channel is in a bad state. Since the belief, $X_t$, represents the probability of the channel being in a good state, the expected reward is given by $X_tR$. If the action of transmitting opportunistically is chosen, $\tau$ fraction of energy unit is spent sensing the channel with the remaining energy being used for transmission if the channel is sensed to be in a good state. In this case, $(1-\tau)R$ bits per time slot are transmitted successfully.  If the channel is sensed to be in a bad state, the transmitter remains silent in the rest of the time slot. The expected reward in this case is $(1-\tau)X_t R$. Finally, if the action of deferring the transmission is taken the transmitter neither senses the channel nor transmits, so the reward is zero.

Define the value function $V(b,\ p)$ as
\begin{align}
V(b,\ p) &= \max_{\pi}V^{\pi}(b,\ p) \nonumber\\
&\text{for all}\ b\in\left[0,\tau,2\tau,\ldots,B_{max}\right]\ \text{and}\ p\in\left[0,\ 1\right].
\end{align}
It is well known that the optimal value of the infinite-horizon expected reward can be achieved by a stationary policy, i.e., there exists a stationary policy $\pi^*$ such that $V(b,\ p) = V^{\pi^*}(b,\ p)$ \cite{puterman}. The value function $V(b,\ p)$ satisfies the Bellman equation
\begin{align}
V(b,\ p) = \max_{A\in\left\{D,O,T\right\}}\left\{V_{A}(b,\ p)\right\},
\end{align}
where $V_{A}(b,\ p)$ is the action-value function, defined as the expected infinite-horizon discounted reward acquired by taking action $A$ when the state is $(b,\ p)$, and is given by
\begin{align}
V_{A}(b,\ p)=&R((b,\ p),A)\nonumber\\
&+\beta\mathds{E}_{(\acute{b},\ \acute{p})}\left[V(\acute{b},\ \acute{p})|S_{0}=(b,\ p),A_0=A\right], \label{actionvalue}
\end{align}
where $(\acute{b},\ \acute{p})$ denotes the next state when action $A$ is chosen at state $S_0=(b,\ p)$. The expectation in (\ref{actionvalue}) is over the distribution of possible next states. In the following, we define and explain the value function $V_{A}(b,\ p)$, and how the system state evolves for each action.


\emph{Deferring the transmission:}
If this action is taken, since there is no transmission, there is no ACK or NAK from the receiver, and thus, the transmitter does not learn the state of the channel. Therefore the next belief is obtained as the probability of finding the channel in a good state given the current belief state. If the transmitter had a belief $X_t=p$ at time slot $t$, after taking action D, its belief at the beginning of the next slot is updated as
\begin{align}
J(p) = \lambda_{0}(1-p)+\lambda_{1}p.
\end{align}
In every time slot, a unit of energy is harvested with probability $q$. Thus, after taking action D, the value function evolves as follows:
\begin{align}
&V_{D}(b,\ p) \nonumber\\
&= \beta\left[qV\left(\min\left\{b+1,B_{max}\right\},\ J(p)\right)+(1-q)V\left(b,\ J(p)\right)\right].
\end{align}
Note that the term $\min\left\{b+1,B_{max}\right\}$ is used to ensure that the battery state does not exceed the battery capacity, $B_{max}$.

\emph{Channel sensing and transmitting opportunistically:}
For this action, two scenarios are possible. If $b\geq 1$ and EH decides to transmit opportunistically, then it consumes $\tau$ fraction of energy to first sense the channel and obtain the current channel state. Based on the outcome of the channel sensing, if the channel is found to be in a good state, $(1-\tau)$ units of energy is used to transmit $(1-\tau)R$ bits per time slot.  Also, the belief state is updated as $\lambda_1$ for the next time slot. 

On the other hand, if the outcome of the channel sensing reveals the channel to be in a bad state, then the transmitter defers its transmission, and saves $(1-\tau)$ units of energy for possible future transmissions. Also, the channel belief is updated as $\lambda_0$ for the next time slot.  Based on the aforementioned discussion, for $b\geq 1$ the evolution of the value function can be written as:
\vspace{-0.2cm}
\begin{align}
V_{O}(b,\ p) =& p\left[(1-\tau)R+ \beta\left(qV(b,\ \lambda_1)+(1-q)V(b-1,\ \lambda_1)\right)\right]\nonumber\\
+&(1-p)\beta\left[qV(\right.\min\left\{b-\tau+1,B_{max}\right\},\ \lambda_0)\nonumber\\
+&(1-q)V(b-\tau,\ \lambda_0)\left.\right].
\end{align}

If $\tau\leq b<1$, then transmission is not possible since transmission requires at least one unit of energy. However, it is still possible to sense the channel, since it only requires $\tau$ fraction of energy. This may happen when transmitter believes that learning the channel state will help its decision in the future. Thus for $\tau\leq b<1$, the value function evolves as:
\begin{align}
V_{O}(b,\ p) =& \beta\left[\right.qpV(b-\tau+1,\ \lambda_1)\nonumber\\
&+q(1-p)V(b-\tau+1,\ \lambda_0)+(1-q)pV(b-\tau,\ \lambda_1)\nonumber\\
&+(1-q)(1-p)V(b-\tau,\ \lambda_0)\left.\right].
\end{align}

\emph{Transmitting without sensing:}
This action can only be chosen if the battery state is greater than or equal to one, i.e., $b\geq 1$\footnote{Note that we are aware that in the generic MDP formulation, in every state, we should have the same set of actions. We can re-define the reward function by assigning $-\infty$ reward for those actions that are not possible to be taken in specific states to account for this issue. For the ease of comprehension, we chose to present the formulation in this manner.}.  Under this action, the transmitter transmits regardless of the actual state of the channel, costing one unit of energy. If the channel is in the good state, $R$ bits per time slot are successfully delivered to the receiver, and the receiver sends back an ACK.  Otherwise, the channel is in the bad state, so the transmission fails, and the receiver sends back a NAK. Meanwhile, the channel is in a good state with probability $p$, i.e., the current belief state, and the belief in the next time slot will be $\lambda_1$. Also the channel is in a bad state with probability $1-p$ and the belief in the next time slot will be $\lambda_0$. Hence, the value function evolves as:
\begin{align}
V_{T}(b,\ p) =& p\left[ R + \beta\left(qV(b,\ \lambda_1)+(1-q)V(b-1,\ \lambda_1)\right)\right]\nonumber\\
&+(1-p)\beta\left[qV(b,\ \lambda_0)+(1-q)V(b-1,\ \lambda_0)\right]
\end{align}

\section{Structure of The Optimal Policy}
\label{sec:StructureOfTheOptimalPolicy}

In this section, we prove that the optimal policy has a threshold type structure on the belief state. First, we need to prove some of the properties of the value function. We begin with establishing the convexity of the optimal value function with respect to the belief state.

\begin{lemma}
\label{thm:convex}
For any given $b\geq 0$, V(b,\ p) is convex in $p$.
\end{lemma} 
\begin{proof}
The proof is given in \cite{mehdi2}.
\end{proof}

In the following lemma, we show that the value function is a non decreasing function of battery state, $b$. This lemma provides the intuition why deferring or sensing actions are advantageous in some states. The incentive of taking these actions is that the value function transitions into higher values without consuming any energy. Moreover, it states that the value function is also non-decreasing with respect to the belief state, $p$. 
\begin{lemma}
\label{thm:nondecreasing-battery}
	 Given any belief $p$, $V(b_1,p)\geq V(b_0,p)$ when $b_1>b_0$. Moreover, for a fixed battery state $b$, if $p_{1}>p_{0}$ then $V(b,p_{1})\geq V(b,p_{0})$. 
\end{lemma}
\begin{proof}
The proof is given in \cite{mehdi2}.
\end{proof}

Finally, Theorem \ref{thm:threshold} below shows that the optimal solution of the problem is a threshold policy with two or three thresholds depending on the system parameters. The threshold values depend on the state of the battery.

\begin{theorem}
\label{thm:threshold}
For any $p\in[0,\ 1]$ and $b\geq 0$, there exists thresholds $0\leq\rho_1(b)\leq\rho_2(b)\leq\rho_3(b)\leq 1$, all of which are functions of the battery state $b$, such that, for $b\geq 1$
\begin{align}
\pi^{*}(b,\ p)
&= \left\{
\begin{array}{rl}
D, & \text{if }\ 0\leq p\leq\rho_1(b)\ \text{or}\ \rho_2(b)\leq p\leq\rho_3(b)\\
O, & \text{if }\ \rho_1(b)\leq p\leq\rho_2(b),\\
T, & \text{if }\ \rho_3(b)\leq p\leq 1,\\\label{thresha}
\end{array} \right.
\end{align}
and for $\tau\leq b<1$,
\begin{align}
\pi^{*}(b,p)
&= \left\{
\begin{array}{rl}
D, & \text{if }\ 0\leq p\leq\rho_1(b)\ \text{or}\ \rho_2(b)\leq p\leq 1,\\
O, & \text{if }\ \rho_1(b)\leq p\leq\rho_2(b). \label{threshb}
\end{array} \right.
\end{align}
\end{theorem}
\begin{proof}
The detailed proof of the theorem is given in \cite{mehdi2}.
\end{proof}
Theorem \ref{thm:threshold} proves that at any battery state $b\geq 1$, at most three threshold values are sufficient to characterize the optimal policy; whereas two thresholds suffice for $0\leq b<1$. However the optimal policy can even be simpler for some battery states and some instances of the problem as it is possible to have $\rho_2(b) = \rho_3(b)$, or even $\rho_1(b) = \rho_2(b) = \rho_3(b)$.

\section{Numerical Results}
\label{num}
In this section we use numerical techniques to characterize the optimal policy, and evaluate its performance. We utilize the value iteration algorithm to calculate the optimal value function. We numerically identify the thresholds for the optimal policy for different scenarios. We also evaluate the performance of the optimal policy, and compare it with some alternative policies in terms of throughput. 
\subsection{Optimal policy evaluation}

In the following, we assume that $B_{max}=5$, $\tau=0.2$, $\beta=0.98$, $\lambda_1=0.9$, $\lambda_0=0.6$, $R = 3$ and $q=0.1$. The optimal policy is evaluated using the value iteration algorithm. In Fig. \ref{5region} each state $(b,\ p)$ is illustrated with a different color corresponding to the optimal policy at that state. In the figure, the areas highlighted with blue color correspond to those states at which deferring the transmission is optimal, green areas correspond to the states at which transmitting opportunistically is optimal, and finally yellow areas correspond to the states for which transmitting without sensing is optimal. As seen in Fig. \ref{5region} any of the three policies (one, two, or three threshold policies) may be optimal depending on the level of the battery state. For example, when the battery state is $b=2$, one-threshold policy is optimal. The transmitter defers transmission up to a belief of state of $p=0.8$ and starts transmitting without sensing beyond this value. For no value of the belief state it opts for sensing the channel. On the other hand, when the battery state is $3.8$, two-threshold policy is optimal, and when the battery state is $2.8$, three-threshold policy is optimal. Considering the low probability of energy arrivals ($q=0.1$) and the relative high cost of sensing ($\tau=0.2$), it is interesting to notice that the transmitter senses the channel even when its battery state is below the transmission threshold, i.e., $b<1$. 

\begin{figure}[ht]
  \centering
    \includegraphics[scale=.14]{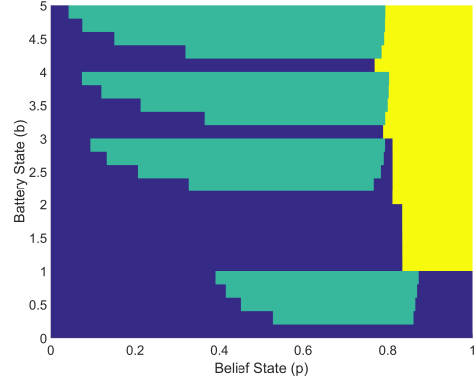}
		  \caption{Optimal thresholds for taking the actions D (blue), O (green), T (yellow) for $B_{max}=5$, $\tau=0.2$, $\beta=0.98$, $\lambda_1=0.9$, $\lambda_0=0.6$, $R = 3$ and $q=0.1$.}
			\label{5region}
\end{figure}

Next, we investigate the effect of the sensing cost, $\tau$, on the optimal policy. To illustrate this effect, we choose the system parameters as before, but increase the sensing cost from $\tau=0.2$ to $\tau=0.5$. Optimal action regions for this setup are shown in Fig. \ref{taucompare}. By comparing Fig. \ref{5region} and Fig. \ref{taucompare}, it is evident that a higher cost of sensing results in less incentive for sensing the channel. We observe in Fig. \ref{taucompare} that the green area has shrunk almost to nothing, i.e, the transmitter is more likely to take a risk and transmit without sensing, or defer its  transmission, when sensing consumes more energy.

\begin{figure}[ht]
  \centering
    \includegraphics[scale=.5]{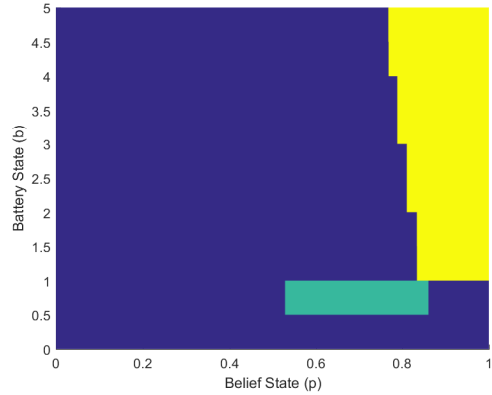}
		  \caption{Optimal thresholds for taking the actions D (blue), O (green), T (yellow) for $B_{max}=5$, $\tau=0.5$, $\beta=0.98$, $\lambda_1=0.9$, $\lambda_0=0.6$, $R = 3$ and $q=0.1$.}
			\label{taucompare}
\end{figure}
\vspace{-0.35cm}
\subsection{Throughput performance}
In this section, we compare the performance of the optimal policy with two alternative policies, namely a greedy policy and a single-threshold policy. In the greedy policy, the transmitter transmits whenever it has energy in its battery. In the single-threshold policy there are only two actions: defer (D) or transmit (T). We optimize the threshold corresponding to each battery state for the single-threshold policy using the value iteration algorithm. By choosing the parameters $B_{max}=5$, $\tau=0.1$, $\beta=0.999$, $\lambda_1=0.7$, $\lambda_0=0.2$, $R = 2$, the  throughput achieved by these three policies are plotted in Fig. \ref{3} with respect to the EH rate $q$. 

As expected we observe that the greedy policy performs the worst as it does not exploit the transmitter's knowledge about the state of the channel. We can see that by simply exploiting  the ACK/NACK feedback from the receiver, it is possible to achieve a higher throughput than the greedy policy for all values of the EH rate. On the other hand, by further introducing the channel sensing action the throughput of the system is substantially increased. The improvement is particularly higher for the mid-range of $q$ values, for which the transmitter benefits more from the flexibility offered by three actions.

\vspace{-0.35cm}

\begin{figure}[ht]
  \centering
    \includegraphics[scale=.5]{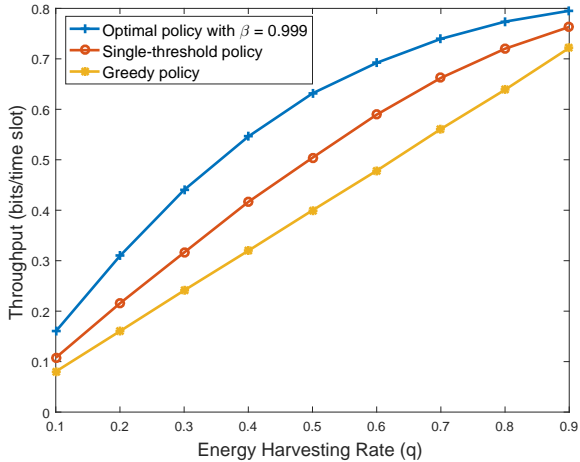}
		  \caption{Throughput comparison among the optimal, greedy, and single-threshold policies as a function of the EH rate, $q$. }
			\label{3}
\end{figure}
\vspace{-0.35cm}
\section{Conclusions and Future Work}
\label{concl}
In this work we considered an EH transmitter equipped with a battery, operating over a time varying finite-capacity wireless channel with memory, modeled as a Gilbert-Elliot channel. The transmitter receives ACK/NACK feedback after each transmission, which can be used to track the channel. We further consider channel sensing, which the transmitter can use to learn the current channel state at a certain energy and time cost. Therefore, at the beginning of each time slot, the transmitter has three possible actions to maximize the total expected discounted number of bits transmitted over an infinite time horizon: i) deferring the transmission to save its energy for future use, ii) transmitting at a rate of $R$ bits per time slot, and iii) sensing the channel to reveal the current channel state by consuming a portion of its energy and time, followed by transmission at a reduced rate consuming the remainder of the energy unit, only if the channel is in the good state. We formulated the problem as a POMDP, which is then converted into a MDP with continuous state space by introducing a belief parameter for the channel state. Then we proved that the optimal policy is a threshold policy, where the threshold values on the belief parameter depends on the battery state. We find the optimal threshold values numerically using the value iteration algorithm. In terms of throughput, we compared the optimal policy to the alternative policies, the greedy policy and a single-threshold policy which does not have channel sensing capability. We have shown through simulations that the channel sensing capability improves the performance significantly, thanks to the increased adaptability to the channel conditions it provides. For future studies, we will consider the case where the sensing is not perfect. Another interesting problem is to consider the case in which the EH transmitter has the option to choose the duration of the sensing which determines its accuracy.

\bibliographystyle{unsrt} 

\bibliography{Bibliography}
\end{document}